\documentclass[11pt]{article}
\usepackage{fullpage}
\usepackage{authblk}

\usepackage{amsmath,amssymb,mathrsfs,mathtools,amsfonts,amsthm}
\usepackage[square,sort,comma,numbers]{natbib}
\allowdisplaybreaks

\usepackage[final]{pdfpages}
\usepackage{ltablex}

\usepackage{color}
\usepackage{amsmath,amssymb,mathrsfs,mathtools}
\usepackage{upgreek,graphicx}

\newtheorem{theorem}{Theorem}
\newtheorem{lemma}{Lemma}

\title{Average-case Analysis of the Assignment Problem with Independent Preferences}

\author[1]{Yansong Gao}
\author[2]{Jie Zhang\thanks{jie.zhang@soton.ac.uk}}
\affil[1]{Applied Mathematics and Computational Science, University of Pennsylvania}
\affil[2]{Electronics and Computer Science,
University of Southampton}
\date{}

\begin{document}

\maketitle

\begin{abstract}
The fundamental assignment problem is in search of welfare maximization mechanisms to allocate items to agents when the private preferences over indivisible items are provided by self-interested agents. 
The mainstream mechanism \textit{Random Priority} is asymptotically the best mechanism for this purpose, when comparing its welfare  to the optimal social welfare using the canonical \textit{worst-case approximation ratio}. 
Despite its popularity, the efficiency loss indicated by the worst-case ratio does not have a constant bound \cite{DBLP:conf/sagt/Filos-RatsikasF014}.
Recently, \cite{DBLP:conf/mfcs/DengG017} show that when the agents' preferences are drawn from a uniform distribution, its \textit{average-case approximation ratio} is upper bounded by 3.718. 
They left it as an open question of whether a constant ratio holds for general scenarios. 
In this paper, we offer an affirmative answer to this question by showing that the ratio is bounded by $1/\mu$ when the preference values are independent and identically distributed random variables, where $\mu$ is the expectation of the value distribution. This upper bound also improves the upper bound of 3.718 in \cite{DBLP:conf/mfcs/DengG017} for the Uniform distribution. Moreover, under mild conditions, the ratio has a \textit{constant} bound for any independent  random values.
En route to these results, we develop powerful tools to show the insights that in most instances the efficiency loss  is small. 
\end{abstract}


\section{Introduction}
Studies in Mechanism Design focus on designing mechanisms in which truth-telling is a dominant strategy, and so the rational, risk-neutral agents are motivated to play their truth-telling strategies. Subject to this qualitative  constraint, a \textit{truthful mechanism} quantitatively optimizes an objective such as social welfare or revenue. 
Build upon the classical Mechanism Design framework, Algorithmic Mechanism Design \cite{NR99} advances research and has evolved to employing two typical analytic tools in Computer Science. One of these tools imposes polynomial-time implementability on designing mechanisms, and the other evaluates mechanism performances through the lens of the worst-case analysis. 
Typically, the (\textit{worst-case}) \textit{approximation ratio}  assesses to what extent a mechanism approximately optimizes an objective over all possible inputs and it is dominated by the worst-case inputs \cite{PT:09}.
There is now an extensive literature on topics including matching ~\cite{DG:10,DBLP:journals/sigecom/Anshelevich16}, scheduling ~\cite{DBLP:journals/mst/Koutsoupias14,DBLP:journals/algorithmica/ChristodoulouKV09}, facility location \cite{DBLP:conf/aaai/Filos-RatsikasL15}, kidney exchange \cite{DBLP:conf/aaai/HajajDHSS15}, fair division \cite{DBLP:journals/geb/ChenLPP13}, social choice \cite{DBLP:journals/ai/AnshelevichBEPS18}, and auction design \cite{DBLP:conf/sigecom/HartlineR09,DBLP:journals/geb/MualemN08,DBLP:conf/soda/ArcherPTT03}. For a more detailed summary, we refer the reader to \cite{NRTV07}.

In the \textit{assignment problem} (a.k.a., \emph{one-sided matching} or \emph{house allocation} problem), there is a set of agents and a set of items. The agents participate in a mechanism by reporting their private preferences over the items. The mechanism then assigns  items to agents, according to a pre-defined allocation function.
The folklore mechanism \emph{Random Priority} (a.k.a., \emph{Random Serial Dictatorship}) is mainstream because it satisfies appealing properties including anonymity, truthfulness, and ex-post Pareto efficiency. In addition, there exists no mechanism that is ex-ante Pareto efficient while keeping the first two desired properties \cite{ZHOU:90}.
Random Priority and its variants have notable practical applications.
For example, they are used in United States Naval Academy placement \cite{RothSotomayorNavy}, social or government-subsidised housing \cite{RePEc:eee:jetheo:v:88:y:1999:i:2:p:233-260}, graduate housing allocation at a large number of universities \cite{RePEc:ecm:emetrp:v:66:y:1998:i:3:p:689-702}, and high school student assignment  in New York \cite{PathakSethuraman2011}.

For the purpose of maximizing \emph{social welfare}, i.e., the sum of all agents' utilities, Random Priority can only achieve a $\Theta{(1/\sqrt{n})}$ fraction of the optimal social welfare in the worst case, where $n$ is the number of agents and items, and it is asymptotically the best amongst all truthful mechanisms \cite{DBLP:conf/sagt/Filos-RatsikasF014}. This negative result was considered a cautionary tale discouraging the wide applications of Random Priority. Fortunately, the smoothed analysis and average-case analysis mitigate it by revealing positive results on social welfare approximation.
In \cite{DBLP:conf/mfcs/DengG017}, the authors show that when the worst-case inputs are subject to small random noise, Random Priority attains social welfare within a constant factor of the optimal welfare. In addition, 
when agents' preferences are drawn from a uniform distribution, on average, the optimal social welfare is no more than $1+e$ times of the social welfare attainable by Random Priority, where $e$ is the Euler's number. Therefore, on average, the efficiency loss is small. The primal question that remains open is: Does the constant average-case ratio result hold for  general probability distributions? In this paper, we partially answer this question by  showing the following results.
\begin{itemize}
\item When agents' preference values are independent and identically distributed random variables, the average-case approximation ratio is upper bounded by $1/\mu$, where $\mu$ is the expectation of these random variables.
\item When agents' preferences values are independent but not necessarily identically distributed random variables, under mild conditions, the average-case approximation ratio is upper bounded by a constant.
\end{itemize}
\noindent
We note that the $1/\mu$ upper bound for the i.i.d. case improves the  3.718 upper bound for the Uniform distribution $\mathrm{U}[0,1]$, as its expectation $\mu=\frac{1}{2}$ and $1/\mu$ yields an  upper bound of 2. Taken together, these results further pin down the wide-applicability of the Random Priority mechanism, for which the worst-case analysis is insufficient to manifest.

To achieve these results, we employ the Central Limit Theorems and carefully calibrate the values of some parameters in a few building blocks  to ascertain the rate that the properly normalized sum of preference values converges to a Normal distribution. By successfully compositing these building blocks together, we bound the average-case ratios.


\subsection{Related Work}
The one-sided matching problem was introduced in \cite{HZ:79} and has been studied extensively ever since. Over the years, several different mechanisms have been proposed with various desirable properties related to truthfulness, fairness and economic efficiency with Probabilistic Serial \cite{AS:13,BCK:11,BM:01,katta2006solution}  and Random Priority \cite{RePEc:ecm:emetrp:v:66:y:1998:i:3:p:689-702,Sv99,DBLP:conf/atal/ChristodoulouFF16,DBLP:conf/wine/AzizBB13}  being the two prominent examples. In the indivisible goods setting, the Top Trading Cycles (TTC) method is well-studied and generalized to investigate various problems. In particular, \cite{RePEc:ecm:emetrp:v:66:y:1998:i:3:p:689-702} proposed an adaptation of the TTC method and established an equivalence between the adapted mechanism and Random Priority. \cite{Kesten2009} proposed several extensions of these popular mechanisms and presented an equivalence result between those mechanisms in terms of economic efficiency. In the presence of incentives, the assignment problem  was extensively studied in Computer Science and Economics over the years \cite{ZHOU:90,DG:10,mennle2014axiomatic}.  We refer the interested reader to  surveys  \cite{AS:13,SU:11}. 
\cite{BCK:11} studied the approximation ratio of matching mechanisms, when the objective is maximization of ordinal social welfare, a notion of efficiency that they define based solely on ordinal information. 

Under the average-case analysis framework,  \cite{DBLP:conf/aaai/Zhang18} tackled the scheduling unrelated machines problem  and showed that the average-case approximation ratio of the mechanism devised in \cite{DBLP:journals/mst/Koutsoupias14} is upper bounded by a constant, when the machines' costs follow any independent and identical distribution. 

A similar notion but fundamentally different to our approach exists in Bayesian analysis \cite{DBLP:journals/sigecom/ChawlaS14,DBLP:conf/stoc/HartlineL10}. We leave a more detailed discussion of the two notions to the next section after formally present the definition of the average-case approximation ratio,  to highlight the comparison.



\section{Preliminaries}
We study the one-sided matching problem that consists of $n$ agents and $n$ indivisible  items.
The agents are endowed with von Neumann - Morgenstern utilities over the items. Throughout the paper, we denote  the utility derived by agent $i$ on obtaining a unit amount of item $j$ by $a_{ij}$. In particular, following the classical literature \cite{ZHOU:90,Barbera10}, agents' preferences $a_{ij}$ are represented by \emph{unit-range}  values. That is,  with normalization, any agent $i$'s valuation on its most preferred item is 1, i.e., $\max_j \{a_{ij}\}=1$, and its valuation on the least preferred item is 0, i.e., $\min_j \{a_{ij}\}=0$. However, note that our model would be more general and some calculations would be cleaner if we drop these constraints but only require that $0\le a_{ij} \le 1$, $\forall i, j \in [n]$.  A \textit{valuation profile} (or interchangeably, an \textit{instance}) of agents' preferences can be represented by a matrix $\mathrm{A}=[a_{ij}]_{n \times n}$, where the row vector $(a_{i1},\ldots,a_{in})$ indicates the valuation of agent $i$'s preference.  

A matching mechanism collects agents' preference valuations and output an assignment of items to them. Denote a matching assignment by a matrix $\mathrm{X}=[x_{ij}]_{n \times n}$, where $x_{ij}$ indicates the probability of  agent $i$ receiving item $j$. So, a mechanism is a mapping from the input instance $\mathrm{A}$ to an output allocation $\mathrm{X}$. Since there is an equal number of agents and items, $\mathrm{X}$ will be a doubly stochastic matrix. According to the Birkhoff - von Neumann Theorem, every doubly stochastic matrix can be decomposed into a convex combination of some permutation matrices. Therefore, any probabilistic allocation $\mathrm{X}$ can be interpreted as a convex combination of a set of deterministic allocations. We denote  the set of all possible instances by $\mathcal{A}$ and denote the set of all possible allocation by $\mathcal{X}$. Given a mechanism $\mathrm{M}$ and a valuation profile $\mathrm{A}\in \mathcal{A}$, as well as its allocation $\mathrm{X}(\mathrm{A})\in \mathcal{X}$, we denote the expected utility of agent $i$ by  $u_i (\mathrm{X}(\mathrm{A})) = \sum_{j} a_{ij}x_{ij}$ and denote the expected social welfare by $SW_{\mathrm{M}}(\mathrm{X}(\mathrm{A})) = \sum_{i} u_i(\mathrm{X}(\mathrm{A}))$. When the context is clear, we drop the allocation notation and simply refer them by $u_i(\mathrm{A})$ and $SW_{\mathrm{M}}(\mathrm{A})$.  

In Mechanism Design, agents are self-interested and  may misreport their values if that results in a better matching (from their perspective). We are interested in \textit{truthful mechanisms}, under which agents cannot improve their utilities by misreporting. Formally, $u_i(\mathbf{a}_i, \mathbf{a}_{-i}) \ge u_i(\mathbf{a}'_i, \mathbf{a}_{-i}), \forall i$, where $\mathbf{a}_i$ is agent $i$'s true valuations, $\mathbf{a}_{-i}$ is other agents' valuations, and $\mathbf{a}'_i$ is any possible misreported valuation from agent $i$. The  mechanism which is the focus of  this paper, \textit{Random Priority} (RP), is a truthful mechanism. It fixes an ordering of the agents uniformly at random and then allocates them their most preferred item from the set of available items based on this ordering.


The canonical measure of efficiency loss due to the restriction to the class of truthful mechanisms, compared to the optimal social welfare, is the  \emph{worst-case approximation ratio},
\begin{equation*}
r_{\text{worst}}(\mathrm{M}) = \sup_{\mathrm{A} \in \mathcal{A}} \frac{SW_{\mathrm{OPT}}(\mathrm{A})}{SW_{\mathrm{M}}(\mathrm{A})},
\end{equation*}
where $SW_{\mathrm{OPT}}(\mathrm{A})= \max_{\mathrm{X} \in \mathcal{X}}\sum_{i=1}^{n}u_i(\mathrm{X})$ is the optimal social welfare which is equal to the value of the maximum weight matching between agents and items. 
It is shown in \cite{DBLP:conf/sagt/Filos-RatsikasF014} that Random Priority achieves the matching approximation ratio bound of $\Theta(\sqrt{n})$.
The \emph{average-case approximation ratio} of a truthful mechanism $\mathrm{M}$ is the \textit{expectation of the ratio} of the optimal social welfare to the social welfare attained by  mechanism $\mathrm{M}$. That is,
\begin{equation*}
r_{\text{average}}(\mathrm{M}) = \mathop{\mathop{\mathrm{E}}}_{a_{ij}\sim \mathrm{D}}\left[ \frac{SW_{\mathrm{OPT}}(\mathrm{A})}{SW_{\mathrm{M}}(\mathrm{A})} \right] ,
\end{equation*}
where the valuation variable $a_{ij}$ follows a  distribution $\mathrm{D}$. This notion of average-case analysis is a pointwise division that is in the same manner as the worst-case ratio $r_{\text{worst}}(\mathrm{M})$ and the smoothed ratio studied in \cite{DBLP:conf/mfcs/DengG017}. When randomly draw an instance from a distribution, it informs us the expected value of how far is the social welfare attainable by a truthful mechanism M on the instance compared to the optimal social welfare \textit{on the same instance}. 

In Bayesian mechanism design,  the dominant approach  is the \emph{ratio of expectations}, defined as $\mathrm{E} \left[ SW_{\mathrm{OPT}}(\mathrm{A}) \right] \le r \cdot \mathrm{E} \left[ SW_{\mathrm{M}}(\mathrm{A}) \right]$. When one's interest is the expected social welfare of a mechanism over all possible inputs compared to the expected optimal social welfare, rather than a pointwise comparison over the same instance, the ratio of expectations would be more appropriate. In most cases, it is more tractable, as it calculates two expectations separately.




\section{Independent and Identically Distributed Random Values}
Subject to the unit-range normalization, there are two preset values (0 and 1, respectively) in each row of a preference matrix $\mathrm{A}$.
Let $S$ be the set of indices such that their corresponding entries in $\mathrm{A}$ are preset. 
That is, $S=\{(i, j) :\ a_{ij}=0 \,\,\  \text{or} \,\,\  a_{ij}=1 \}$.
Obviously, $|S| = 2n$, where $|*|$ denotes the cardinality of a set\footnote{It does not matter which $2n$ of these  valuations $a_{ij}$ are preset. Our proofs hold for any choice of these $2n$ entries in $\mathrm{A}$.}. 
In this section, we allow the remaining $n^2-2n$ values to be  independent and identically distributed random variables following a distribution $\mathrm{D}$ with expectation $\mu$ and variance $\sigma^2$, i.e., $a_{ij} \sim \mathrm{D}[0,1]$, $\mathrm{E}_D [a_{ij}]=\mu, \mathrm{Var}[a_{ij}]=\sigma^2, \forall i,j.$ 
We will show that the average-case approximation ratio of RP is upper bounded by a constant, for any distribution $\mathrm{D}[0,1]$.
In order to prove our main results, we partition the average-case ratio into two cases according to the social welfare attainable by RP. We carefully choose a scalar such that in one case, the probability that the social welfare attainable by RP  is smaller than the scalar is asymptotically small (even if there is a large efficiency loss in this case), and in the other case, the social welfare attainable by RP is asymptotically close to the optimal social welfare. En route to complying with these two conditions simultaneously, we employ the Central Limit Theorem and control the rate that the normalized sum of individual values converges to a normal distribution by carefully calibrating a scalar.

One of the tools that we use to control the convergence rate is the Berry-Esseen Theorem \cite{RickDurrett}.


\begin{theorem}\label{BE}
(Berry-Esseen) For $n\geq 1$, let $X_{1},\cdots,X_{n}$ be i.i.d. random variables such that $\mathrm{E}[X_1]=0$ and $\mathrm{D}[X_{1}]=\sigma^{2}$. Denote $S_{n}=\frac{X_{1}+\cdots+X_{n}}{\sigma\sqrt{n}}$ and $F_{n}(x)=\Pr\{S_{n}\leq x\}$. Then
\begin{align*}
\sup_{x} |F_{n}(x)-\Phi(x)|\leq \frac{C \cdot \mathrm{E}[|X_{1}|]^{3}}{\sigma^{3}\sqrt{n}},
\end{align*}
where $\Phi(x)$ is the cumulative distribution function of the standard Normal distribution, and $0.409 < C<0.475 $.  
\end{theorem}

Now we will use the Berry-Esseen Theorem and the Central Limit Theorem to prove a key lemma. 

\begin{lemma}\label{SWsmaller}
For a given preference instance $A$, the probability that the social welfare $SW_{RP}(A)$ attainable by Random Priority is less than $\lambda$, is bounded by the following inequality.
\begin{align*}
\Pr \left\{ SW_{RP}(A) \le \lambda \right\} \le \frac{1}{2n \sqrt{\pi \ln n}} + \frac{C}{\sigma^3 \sqrt{n(n-2)}},
\end{align*}
where $C$ is the constant in  the Berry-Esseen Theorem and $\lambda=1+\mu(n-2) - \sigma \sqrt{\frac{2(n-2)}{n} \ln n}$.
\end{lemma}

\begin{proof}
Firstly, for any preference instance $A$, the social welfare attainable by Random Priority is lower bounded by the following inequality.
\begin{align*}
SW_{RP}(\mathrm{A}) &\ge \frac{1}{n} \sum_{i=1}^{n}  \sum_{j=1}^{n} a_{ij} \\
&= \frac{1}{n} \sum_{(i,j) \in S}  a_{ij} +  \frac{1}{n} \sum_{(i,j) \not\in S} a_{ij} \\
&= 1 + \frac{1}{n} \sum_{(i,j) \not\in S}  a_{ij} .
\end{align*}
Therefore,
\begin{align}
&\Pr \left\{ SW_{RP}(A) \le \lambda \right\} \le \Pr \left\{ 1 + \frac{1}{n} \sum_{(i,j) \not\in S}  a_{ij}  \le \lambda \right\}  \nonumber \\
&= \Pr \left\{ \sum_{(i,j) \not\in S}  a_{ij}  \le \mu \cdot n(n-2) - \sigma \sqrt{2n(n-2) \ln n}  \right\}
\end{align}

Secondly, since $a_{ij}$ are i.i.d. random variables and $\mathrm{E} [a_{ij}]=\mu, \mathrm{Var}[a_{ij}]=\sigma^2, \forall i,j$, we know that $a_{ij} - \mu$ are i.i.d. random variables such that $\mathrm{E} [a_{ij} - \mu]=0, \mathrm{Var}[a_{ij} - \mu]=\sigma^2, \forall i,j.$ Let
\begin{align*}
T_{n(n-2)}=\frac{ \sum_{(i,j) \not\in S} (a_{ij} - \mu) }{\sigma\sqrt{n(n-2)}}
\end{align*}
and $G_{n(n-2)}(x)=\Pr\{T_{n(n-2)}\leq x\}$.

Following Theorem \ref{BE}, the random variables $T_{n(n-2)}$  converge in distribution to the standard Normal distribution $N(0, 1)$. That is, $T_{n(n-2)} \xrightarrow{d}  N(0, 1)$. In addition, the rate of convergence is bounded by
\begin{align*}
\sup_{x}|G_{n(n-2)}(x)-\Phi(x)| \leq \frac{C \cdot \mathrm{E} \left[ |a_{ij}-\mu| \right]^{3}}{\sigma^{3}\sqrt{n(n-2)}} \leq \frac{C}{\sigma^{3}\sqrt{n(n-2)}},
\end{align*}
where $(i,j) \notin S$, and the second inequality holds because $0 < a_{ij} < 1$ and $0 < \mu < 1$.
Therefore,
\begin{align}
&\left|\Pr\left\{ T_{n(n-2)}  \le x \right\} - \Phi(x)  \right|  \leq  \frac{C}{\sigma^{3}\sqrt{n(n-2)}}, \,\,\,\,\   \forall x. \nonumber \\
&\Rightarrow  \left|\Pr\left\{ \frac{ \sum_{(i,j) \not\in S} (a_{ij} - \mu) }{\sigma\sqrt{n(n-2)}}   \le x \right\} - \Phi(x)  \right|  \leq  \frac{C}{\sigma^{3}\sqrt{n(n-2)}}   \nonumber \\
&\Rightarrow \Pr \left\{ \sum_{(i,j) \not\in S} a_{ij} \le \mu \cdot n(n-2) + x \cdot \sigma \sqrt{n(n-2)} \right\} \nonumber \\
&\le \int_{-\infty}^{x} \frac{1}{\sqrt{2\pi}} e^{-\frac{t^2}{2}} dt + \frac{C}{\sigma^{3}\sqrt{n(n-2)}}. 
\end{align}
Let $x=-\sqrt{2\ln n}$. Obviously, $-x \rightarrow +\infty$ when $n$ approaches infinity. Following (2), we get that
\begin{align}
&\Pr \left\{ \sum_{(i,j) \not\in S} a_{ij} \le \mu \cdot n(n-2) - \sigma \sqrt{2\ln n} \cdot  \sqrt{n(n-2)} \right\} \nonumber \\
&\le \int_{-\infty}^{-\sqrt{2\ln n}} \frac{1}{\sqrt{2\pi}} e^{-\frac{t^2}{2}} dt + \frac{C}{\sigma^{3}\sqrt{n(n-2)}} \nonumber \\
&= \int_{\sqrt{2\ln n}}^{+\infty} \frac{1}{\sqrt{2\pi}} e^{-\frac{t^2}{2}} dt + \frac{C}{\sigma^{3}\sqrt{n(n-2)}} \nonumber \\
&\le \frac{1}{\sqrt{2\pi}\sqrt{2\ln n}} e^{-\ln n}  + \frac{C}{\sigma^{3}\sqrt{n(n-2)}} \nonumber  \\
&= \frac{1}{2n \sqrt{\pi \ln n}} + \frac{C}{\sigma^3 \sqrt{n(n-2)}}.
\end{align}

Combining (1) and (3),  we prove that
\begin{align*}
\Pr \left\{ SW_{RP}(A) \le \lambda \right\} \le \frac{1}{2n \sqrt{\pi \ln n}} + \frac{C}{\sigma^3 \sqrt{n(n-2)}}.
\end{align*}
\end{proof}

We shall see that $\lambda$ was chosen such that $\lambda \sim \Theta(n)$, and $\Pr \left\{ SW_{RP}(A) \le \lambda \right\} \sim O(\frac{1}{n})$. 
After establishing these building blocks, we are ready to prove the main theorem of this section.

\begin{theorem}
The average-case approximation ratio of Random Priority, when agents' preferences are independent and identically distributed random variables, is upper bounded by the constant $1/ \mu$, where $\mu$ is the expectation of these random variables.
\end{theorem}
\begin{proof}
To calculate the expectation of the ratio of the optimal social welfare to the social welfare attained by RP, we partition the value space of the ratio into two cases according to the threshold parameter $\lambda$. In each case, we multiply the ratio by the probability that the case occurs. That is, 
\begin{align*}
r_{\text{average}}(\mathrm{RP}) &= \mathop{\mathop{\mathrm{E}}}_{a_{ij}\sim \mathrm{D}}\left[ \frac{SW_{\mathrm{OPT}}(\mathrm{A})}{SW_{\mathrm{RP}}(\mathrm{A})} \right] \\
&\leq \Pr\{SW_{RP}(\mathrm{A})> \lambda \}\cdot \frac{n}{\lambda}  + \Pr\{SW_{RP}(\mathrm{A}) \leq \lambda\}\cdot r_{\text{worst}}
\end{align*}
Here we plugged in the fact that $SW_{\mathrm{OPT}}(\mathrm{A}) \le n$, for any instance $\mathrm{A}$.

In addition, as shown in \cite{DBLP:conf/sagt/Filos-RatsikasF014}, $r_{\text{worst}}  \sim \Theta(\sqrt{n})$, there exists a constant $c_1$, such that $r_{\text{worst}} \le c_1 \cdot \sqrt{n}$, for sufficiently large $n$. Also, it is obvious that  $\Pr\{SW_{RP}(\mathrm{A})> \lambda \} \le 1$. Now, we  plug in these fact and the results established in Lemma 1.

\begin{align*}
r_{\text{average}}(\mathrm{RP})  &\le 1 \cdot \frac{n}{1+\mu(n-2) - \sigma \sqrt{\frac{2(n-2)}{n}\ln n} }   + \left[ \frac{1}{2n \sqrt{\pi \ln n}} + \frac{C}{\sigma^3 \sqrt{n(n-2)}} \right] \cdot c_1 \sqrt{n} \\
& \le \frac{1}{\mu} \left( \frac{1}{1-\frac{2}{n} + \frac{1}{\mu n} - \frac{\sigma \sqrt{2 \ln n}}{\mu n}} \right)  + \frac{c_1}{2\sqrt{\pi n \ln n}} + \frac{C \cdot c_1}{\sigma^3 \sqrt{n-2}}
\end{align*}
Since $\frac{1}{1-|t|} < 1+2|t|$, when $|t|< \frac{1}{2}$. As long as $\left| \frac{2}{n} - \frac{1}{\mu n} + \frac{\sigma \sqrt{2 \ln n}}{\mu n} \right| < \frac{1}{2}$, we conclude that
\begin{align*}
r_{\text{average}}(\mathrm{RP})  
& \le \frac{1}{\mu} \left( 1 + 2 \left| \frac{2}{n} - \frac{1}{\mu n} + \frac{\sigma \sqrt{2 \ln n}}{\mu n} \right| \right)  + \frac{c_1}{2\sqrt{\pi n \ln n}} + \frac{C \cdot c_1}{\sigma^3 \sqrt{n-2}} \\
&= \frac{1}{\mu} +  \frac{2}{\mu} \left| \frac{2}{n} - \frac{1}{\mu n} + \frac{\sigma \sqrt{2 \ln n}}{\mu n} \right|  + \frac{c_1}{2\sqrt{\pi n \ln n}} + \frac{C \cdot c_1}{\sigma^3 \sqrt{n-2}} \\
&\rightarrow\frac{1}{\mu}.
\end{align*}
\end{proof}


\section{Independent but not Necessarily Identical Random  Values}
In this section, we further generalize the results in Section 3 to the scenario that the $n^2-2n$ elements in $\left\{ a_{ij} | (i, j) \notin S \right\}$ are independent random values but are not restricted to following an identical distribution. So, they may have different expectations $\mathrm{E} [a_{ij}]=\mu_{ij}$ and variance $\mathrm{Var}[a_{ij}]=\sigma_{ij}^2, \forall i,j.$ We will show that  under mild conditions, the average-case approximation ratio of RP has a constant  upper bound. 
The primary obstruction in this generalization is the difficulty to appropriately control the rate of convergence of the normalized sum of the random values, in order to obtain a constant upper bound. 
However, we are able to pinpoint a mild condition, such that when the random variables $a_{ij}$ comply with the condition, we can establish similar building blocks to the last section. By carefully calibrating the parameter $\lambda$, we can assemble them in a compatible way to get a constant upper bound.  

Firstly, we identify the following mild conditions:\\
 (i) $\sum_{(i,j)\notin S}\mu_{ij} = \Omega(n^2)$; \,\,\,\ (ii) $\sum_{(i,j)\notin S} \sigma_{ij}^2 = \omega(n)$. 
\\ Note that there are $n^2-2n$ elements in the set $\left\{ a_{ij} | (i, j) \notin S \right\}$, so the first condition simply implies that there would not be many of these $a_{ij}$ whose expectations are asymptotically small. The second condition merely implies that the variances of the valuations are not too small. In other words, the first condition excludes those instances that many agents' preferences  are negligible; the second condition requires their preferences to admit a magnitude  of variation.

Secondly, to control the rate of convergence, we will employ the following theorem which is a refined version of the Berry-Esseen Theorem. It holds for  non-identically  distributed random variables \cite{Esseen1945}. 

\begin{theorem}\label{BE2}
Let $Z_{1},\cdots,Z_{n}$ be independent random variables such that $\mathrm{E}[Z_i]=0$, $\mathrm{D}[Z_{i}]=\sigma_{i}^{2}$, and $\mathrm{E}[|Z_{i}|]^{3} < \infty$. Denote $X_{n}=\frac{Z_{1}+\cdots+Z_{n}}{\sqrt{\sum_{i} \sigma_{i}^2}}$ and $F_{n}(x)=\Pr\{X_{n}\leq x\}$. Then there exists a constant $C'$ such that
\begin{align*}
\sup_{x} |F_{n}(x)-\Phi(x)|\leq \frac{C' \cdot \sum_{i} \mathrm{E}[|Z_{i}|]^{3}}{ \left(\sum_{i} \sigma_{i}^2\right)^{\frac{3}{2}} },
\end{align*}
where $\Phi(x)$ is the CDF of the standard Normal distribution.
\end{theorem}

Now, let $Z_{ij}:=a_{ij}-\mu_{ij}$, where $(i,j)\notin S$. Denote their normalized sum of $Z_{ij}$ and the CDF of$X_{n(n-2)}$ by
\begin{align*}
X_{n(n-2)}:=\frac{\sum_{i,j}Z_{ij}}{ \sqrt{\sum_{i,j} \sigma_{ij}^2} }, \,\,\,\ F_{X_{n(n-2)}}(x)=\Pr(X_{n(n-2)}\leq x) .
\end{align*}

Next, let $\lambda=1+\frac{\sum_{(i,j) \notin S}\mu_{ij}}{n} -  \frac{\sqrt{2\ln n}}{n} \cdot \sqrt{ \sum_{(i,j) \notin S} \sigma_{ij}^2 }$. With this carefully chosen value of $\lambda$, we are able to show the following lemma.
\begin{lemma}\label{SWsmaller}
 For a given preference matrix $A$, the probability that the social welfare $SW_{RP}(A)$ attainable by Random Priority is less than $\lambda$, is bounded by the following inequality.
\begin{align*}
\Pr \left\{ SW_{RP}(A) \le \lambda \right\} \le \frac{1}{2n \sqrt{\pi \ln n}} + \frac{C' }{ \sqrt{\sum_{i,j} \sigma_{ij}^2}} ,
\end{align*}
where $C'$ is the constant in Theorem \ref{BE2}.
\end{lemma}

\begin{proof}
On the one hand, $SW_{RP}(\mathrm{A}) \ge 1 + \frac{1}{n} \sum_{(i,j) \not\in S}  a_{ij}$ implies that
\begin{align}\label{SWsmaller2}
&\Pr \left\{ SW_{RP}(A) \le \lambda \right\} \le \Pr \left\{ 1 + \frac{1}{n} \sum_{(i,j) \not\in S}  a_{ij}  \le \lambda \right\}  \nonumber \\
&= \Pr \left\{ \sum_{(i,j) \not\in S}  a_{ij}  \le \sum_{(i,j) \notin S}\mu_{ij}  - \sqrt{2\ln n} \cdot \sqrt{\sum_{(i,j) \notin S} \sigma_{ij}^2} \right\}
\end{align}

On the other hand, according to Theorem \ref{BE2}, there exists a constant $C'$, such that
\begin{align*}
&\sup_{x}|F_{X_{n(n-2)}}(x)-\Phi(x)|
\leq  \frac{C' \cdot \sum_{i,j} \mathrm{E}[|Z_{ij}|]^{3}}{ \left(\sum_{i,j} \sigma_{ij}^2\right)^{\frac{3}{2}} }.
\end{align*}
Since $|Z_{ij}| \le 1$, we have $\mathrm{E}[|Z_{ij}|]^{3} \le \mathrm{E}[|Z_{ij}|]^{2} = \sigma_{ij}^2$. So, 
\begin{align*}
&\sup_{x}|F_{X_{n(n-2)}}(x)-\Phi(x)|
\leq  \frac{C' }{ \sqrt{\sum_{i,j} \sigma_{ij}^2} }.
\end{align*}
Therefore, $\forall x$,
\begin{align}\label{aij1}
&\left|\Pr\left\{ X_{n(n-2)}  \le x \right\} - \Phi(x)  \right|  \leq   \frac{C' }{ \sqrt{\sum_{i,j} \sigma_{ij}^2} }, \,\,\,\,\   \nonumber \\
&\Rightarrow  \left|\Pr\left\{ \frac{ \sum_{i,j} (a_{ij} - \mu_{ij}) }{\sqrt{\sum_{i,j} \sigma_{ij}^2}} \le x \right\} - \Phi(x)  \right|   \leq \frac{C' }{ \sqrt{\sum_{i,j} \sigma_{ij}^2} }   \nonumber \\
&\Rightarrow \Pr \left\{ \sum_{i,j} a_{ij} \le \sum_{i,j}\mu_{ij}  + x \cdot \sqrt{\sum_{i,j} \sigma_{ij}^2} \right\} \nonumber \\
&\le \int_{-\infty}^{x} \frac{1}{\sqrt{2\pi}} e^{-\frac{t^2}{2}} dt +  \frac{C' }{ \sqrt{\sum_{i,j} \sigma_{ij}^2} } .
\end{align}
Let $x=-\sqrt{2\ln n}$. Obviously, $-x \rightarrow +\infty$ when $n$ approaches infinity. Following (\ref{aij1}), we obtain that
\begin{align}\label{aij2}
&\Pr \left\{ \sum_{i,j}  a_{ij}  \le \sum_{i,j}\mu_{ij}  - \sqrt{2\ln n} \cdot \sqrt{\sum_{i,j} \sigma_{ij}^2} \right\} \nonumber \\
&\le \int_{-\infty}^{-\sqrt{2\ln n}} \frac{1}{\sqrt{2\pi}} e^{-\frac{t^2}{2}} dt + \frac{C' }{ \sqrt{\sum_{i,j} \sigma_{ij}^2}}  \nonumber \\
&\le \frac{1}{\sqrt{2\pi}\sqrt{2\ln n}} e^{-\ln n}  + \frac{C' }{ \sqrt{\sum_{i,j} \sigma_{ij}^2}} \nonumber  \\
&\le \frac{1}{2n \sqrt{\pi \ln n}} + \frac{C' }{ \sqrt{\sum_{i,j} \sigma_{ij}^2}}  .
\end{align}

Combining (\ref{SWsmaller2}) and (\ref{aij2}),  we complete the proof.
\end{proof}

Finally, we are ready to prove the main theorem of this section.

\begin{theorem}
The average-case approximation ratio of Random Priority, when agents' preferences are independent but not necessarily identically distributed random variables, is upper bounded by a constant.
\end{theorem}

\begin{proof}
We partition the value space of the ratio into two cases according to the threshold parameter $\lambda$. In each case, we multiply the ratio by the probability that the case occurs. That is,
\begin{align*}
&r_{\text{average}}(\mathrm{RP}) = \mathop{\mathop{\mathrm{E}}}_{a_{ij}\sim \mathrm{D_{ij}}}\left[ \frac{SW_{\mathrm{OPT}}(\mathrm{A})}{SW_{\mathrm{RP}}(\mathrm{A})} \right] \\
&\leq \Pr\{SW_{RP}(\mathrm{A})> \lambda \}\cdot \frac{n}{\lambda} + \Pr\{SW_{RP}(\mathrm{A}) \leq \lambda \}\cdot c_1 \sqrt{n} 
\end{align*}
We assemble the above building blocks and get that 
\begin{align*}
r_{\text{average}}(\mathrm{RP})  &\le 1 \cdot \frac{n}{1+\frac{\sum_{i,j}\mu_{ij}}{n} -  \frac{\sqrt{2\ln n}}{n} \cdot \sqrt{ \sum_{i,j} \sigma_{ij}^2 }}   + \left( \frac{1}{2n \sqrt{\pi \ln n}} + \frac{C' }{ \sqrt{\sum_{i,j} \sigma_{ij}^2}} \right) \cdot c_1 \sqrt{n} \\
& \le \frac{n^2}{\sum \mu_{i,j}} \left( \frac{1}{1- \left( \sqrt{2 \ln n} \cdot \frac{ \sqrt{\sum \sigma_{ij}^2} }{\sum \mu_{ij} } - \frac{n}{\sum \mu_{ij}} \right) } \right)  + \frac{c_1}{2\sqrt{\pi n \ln n}} + \frac{C' \cdot c_1 \sqrt{n}}{\sqrt{ \sum \sigma_{ij}^2} }
\end{align*}
According to the first condition, there  $\exists \underline{\mu}$, $\exists N_0$, such that $\forall n>N_0$, $\frac{n^2}{\sum \mu_{ij}} \le \frac{1}{\underline{\mu}}$. Note that $0<a_{ij}<1$ and $0<\mu_{ij}<1$, so $\sum \mu_{ij} < n^2$, hence  $\frac{1}{\underline{\mu}}>1$. We also know that $\sum \sigma_{ij}^2 < n^2$. Together with the first condition, we have that $\sqrt{2 \ln n} \cdot \frac{ \sqrt{\sum \sigma_{ij}^2} }{\sum \mu_{ij} } \rightarrow 0$ and $\frac{n}{\sum \mu_{ij}} \rightarrow 0$, when $n \rightarrow \infty$. Therefore, there $\exists N_1$, such that when $n>N_1$, $\left| \frac{\sqrt{2 \ln n} \sqrt{\sum \sigma_{ij}^2}}{\sum \mu_{ij}} - \frac{n}{\sum \mu_{ij}} \right| < \frac{1}{2}$. Hence, when $n>\max \{N_0, N_1\}$, we have that 
\begin{align*}
r_{\text{average}}(\mathrm{RP})  
& \le \frac{1}{\underline{\mu}} \left( 1 + 2 \left| \frac{\sqrt{2 \ln n} \sqrt{\sum \sigma_{ij}^2}}{\sum \mu_{ij}} - \frac{n}{\sum \mu_{ij}} \right| \right)  + \frac{c_1}{2\sqrt{\pi n \ln n}} + \frac{C' \cdot c_1 \sqrt{n}}{\sqrt{ \sum \sigma_{ij}^2} } \\
&= \frac{1}{\underline{\mu}} +  \frac{2}{\underline{\mu}} \left| \frac{\sqrt{2 \ln n} \sqrt{\sum \sigma_{ij}^2}}{\sum \mu_{ij}} - \frac{n}{\sum \mu_{ij}} \right|   + \frac{c_1}{2\sqrt{\pi n \ln n}} + \frac{C' \cdot c_1 \sqrt{n}}{\sqrt{ \sum \sigma_{ij}^2} } \\
\end{align*}
Because of condition (ii), we know that $ \frac{\sqrt{n}}{\sqrt{ \sum \sigma_{ij}^2} } \rightarrow 0$. In conclusion, 
\begin{align*}
r_{\text{average}}(\mathrm{RP})  \rightarrow \frac{1}{\underline{\mu}} + 0 = \frac{1}{\underline{\mu}}.
\end{align*}
\end{proof}


\section{Conclusion}
This paper extended the average-case analysis in \cite{DBLP:conf/mfcs/DengG017} from a uniform distribution to any independent distribution and showed a constant upper bound of the approximation ratio. The average-case analysis complements classical worst-case analysis when the worst-case performance is insufficient to characterize the performance of a mechanism. Our results further justify the wide-applicability of the Random Priority mechanism.

There are a few technical points we would like to highlight here. Firstly, the techniques presented in this paper are probably applicable to analyzing other mechanisms and other domains. Secondly, there are various generalizations of the Berry-Esseen Theorem, and each of them may be cast to prove similar results in Section 4. In an earlier version of the present paper, we had also independently proved  a version of the convergence rate. The difference in utilizing different versions of these theorems is that they each require a set of conditions to make the rest of the proof work, and the interpretations of those conditions could be different. 

There are a number of problems remain open as well. For example, one may be interested in investigating the average-case ratio of Random Priority in correlated domains. Also, it would be interesting to proving tighter bounds by making more use of the structure of the assignment problem domain.




\bibliographystyle{named}
\bibliography{refs}

\end{document}